\documentclass[11pt]{article}
\usepackage[a4paper,margin=1in]{geometry}
\usepackage{amsmath,amssymb,amsthm}
\usepackage{authblk}
\usepackage{hyperref}
\usepackage{microtype}
\usepackage{lmodern}
\usepackage[T1]{fontenc}
\usepackage[utf8]{inputenc}
\usepackage{xcolor}

\hypersetup{colorlinks=true, linkcolor=black, citecolor=blue, urlcolor=blue}

\title{\textbf{Intrinsic Heisenberg-type lower bounds on spacelike hypersurfaces in general relativity}}

\author{Thomas Sch\"urmann\thanks{Electronic address: \texttt{t.schurmann@icloud.com}}}

\affil{J\"ulich Supercomputing Centre, Research Centre J\"ulich, 52425 Jülich, Germany}
\date{}

\theoremstyle{plain}
\newtheorem{theorem}{Theorem}[section]

\newtheorem{corollary}[theorem]{Corollary}

\theoremstyle{definition}
\newtheorem{assumption}[theorem]{Assumption}
\newtheorem{definition}[theorem]{Definition}

\theoremstyle{remark}
\newtheorem{remark}[theorem]{\textbf{Remark}}

\begin{document}
	\maketitle

\begin{abstract}
	In quantum theory on curved backgrounds, Heisenberg's uncertainty principle is usually discussed in terms of ensemble variances and flat-space commutators. Here we take a different, preparation-based viewpoint tailored to sharp position measurements on spacelike hypersurfaces in general relativity. A projective localization is modeled as a von Neumann--L\"uders projection onto a geodesic ball $B_\Sigma(r)$ of radius $r$ on a Cauchy slice $(\Sigma,h)$, with the post-measurement state described by Dirichlet data. Using DeWitt-type momentum operators adapted to an orthonormal frame, we construct a geometric, coordinate-invariant momentum standard deviation $\sigma_p$ and show that strict confinement to $B_\Sigma(r)$ enforces an intrinsic kinetic-energy floor. The lower bound is set by the first Dirichlet eigenvalue $\lambda_1$ of the Laplace--Beltrami operator on the ball, $\sigma_p \ge \hbar\sqrt{\lambda_1}$, and is manifestly invariant under changes of coordinates and foliation. A variance decomposition separates the contribution of the modulus $|\psi|$ from phase-gradient fluctuations and clarifies how the spectral geometry of $(\Sigma,h)$ controls momentum uncertainty. 
	Assuming only minimal geometric information, weak mean-convexity of the boundary yields a universal, scale-invariant Heisenberg-type product bound, $\sigma_p r \ge \pi\hbar/2$, depending only on the proper radius $r$.	
\end{abstract}

	\paragraph{Keywords.} Heisenberg-type inequality; quantum mechanics on curved spacetime; spacelike hypersurfaces; spectral geometry; von Neumann-Lüders projection; 
	
\section{Introduction}

Historically, Kennard \cite{K27} was the first to adopt the standard deviation as a quantitative measure of uncertainty in quantum mechanics. However, neither he nor Heisenberg \cite{H27,H30} provided an explicit justification for this choice from the standpoint of experimental physics.
In the textbook formulation, Kennard’s choice \cite{K27} is effectively encoded in the ensemble-based variance relation
\begin{equation*}
\sigma_p \sigma_x \;\geq\; \frac{\hbar}{2},
\end{equation*}
which is obtained by measuring position and momentum on two identically $\psi$-prepared subensembles.  

By contrast, single-slit diffraction of particles has long been a paradigmatic illustration of Heisenberg’s uncertainty relations and their significance in quantum measurement \cite{H27,H30}. A beam of particles is sent through a slit of width $\Delta x$ in an otherwise opaque diaphragm; behind the slit, the transmitted wave is diffracted and acquires a spread in transverse momentum that grows as $\Delta x$ is reduced.  This is the standard heuristic picture used to introduce Heisenberg's principle \cite{H27,H30}.  If one regards passage through the slit itself as a preparation procedure, then the relevant quantity is the momentum spread of the transmitted beam.  A finite momentum standard deviation $\sigma_p < \infty$ is obtained only if the wave function of the transmitted state vanishes at the edges of the interval of width $\Delta x$; otherwise $\sigma_p$ is not even defined.  Under this natural regularity condition the optimal sharp inequality reads \cite{TS09,TS23}
\begin{equation*}
	\sigma_p \,\Delta x \;\geq\; \pi \hbar.
\end{equation*}
Mathematically, this bound is nothing but the one-dimensional Wirtinger inequality (\cite{Wirtinger}, Ch.\,1.7), the prototype of the wider class of Poincar\'e inequalities.  Physically, it embodies a different viewpoint from Kennard's: the preparation is implemented by a concrete projection (the slit), and the uncertainty relation links a \emph{geometric} parameter $\Delta x$ to the resulting momentum spread.

These two inequalities,
\begin{equation*}
\sigma_p \sigma_x \;\geq\; \frac{\hbar}{2}
\qquad\text{and}\qquad
\sigma_p \Delta x \;\geq\; \pi\hbar,
\end{equation*}
are therefore operationally inequivalent.  In the Kennard relation \cite{K27}, one first chooses a state $\psi$ and then reads off $\sigma_x$ and $\sigma_p$ from separate ensembles prepared in that same state; the relation constrains the spreads of position and momentum for a fixed $\psi$.  In the single-slit setting, by contrast, the preparation itself is given by a sharp spatial restriction: the incoming state $\psi$ is projected onto the interval of width $\Delta x$ defined by the slit, and the transmitted ensemble is described by the von Neumann--L\"uders update of $\psi$ \cite{vonNeumann1932, Lueders1951}.  The corresponding inequality directly relates the geometric preparation parameter $\Delta x$ of the \emph{projected} state to its momentum standard deviation $\sigma_p$.  Neither inequality can be derived from the other; they capture complementary aspects of the incompatibility of position and momentum.

If one asks how Heisenberg's principle should be formulated in a relativistic setting, this preparation-based point of view becomes particularly important.  In regions where gravity is strong, any attempt to localize a particle (or an event) within a very small spacetime region requires large momenta and therefore high energy.  According to Einstein's field equations, such energy densities curve the spacetime metric and may even lead to horizon formation.  Much of the literature on this interface between quantum mechanics and gravity proceeds by generalizing Kennard's ensemble-based relation: one modifies the canonical commutation relations or the functional form of the position-momentum uncertainty bound, leading to a family of so-called Generalized Uncertainty Principles (GUPs) \cite{Perivolaropoulos:2017rgq,Parsamehr:2024,Skara:2019,BensalemBouaziz:2022,Lawson:2020}.  Structurally, these proposals keep the Kennard picture of preparation and measurement but deform the algebra or the bound.

If one wants spacetime geometry and gravitational backreaction to enter already at the level of the preparation step, it is natural to turn the logic around.  The single-slit experiment suggests that a sharp preparation is most faithfully modeled by a projection onto a spatial region, not by an abstract choice of wave function.  In flat space this region is an interval or a slit; in general relativity, the appropriate analogue is a spatial domain on a Cauchy slice.  The basic idea of this work is therefore to take von Neumann--L\"uders projections onto higher-dimensional spatial regions as the starting point and to let the geometry of those regions control the resulting momentum uncertainty.

Concretely, we work on a Riemannian manifold $(\Sigma,h)$ representing a spacelike hypersurface of a Lorentzian spacetime.  A sharp localization of a quantum state to a bounded domain $\Omega\subset\Sigma$ with Dirichlet boundary conditions ($\psi|_{\partial\Omega}=0$) plays the role of an idealized projective preparation: the probability to find the particle outside $\Omega$ vanishes.  From the spectral point of view, such strict localization has an unavoidable cost.  Once Dirichlet data are imposed, the $L^2$-norm of the gradient of $\psi$ cannot be made arbitrarily small: the intrinsic uncertainty is governed by the first Dirichlet eigenvalue of the Laplace--Beltrami operator on $\Omega$ \cite[Ch.\,3]{Chavel1984}.  By the Rayleigh--Ritz characterization, the first Dirichlet eigenvalue $\lambda_1(\Omega;h)$ is the minimal value of the Dirichlet energy
\begin{equation*}
\int_\Omega \|\nabla_h \psi\|_h^2\,d\mu_h
\end{equation*}
among all normalized wave functions $\psi\in H^1_0(\Omega)$.  Thus strict localization to $\Omega$ enforces a nonzero kinetic-energy floor set by $\lambda_1$.  Since the geometric momentum variance defined in Section~\ref{sec:var} is directly related to this kinetic energy, $\lambda_1$ emerges as the intrinsic spectral invariant that fixes the minimal momentum uncertainty compatible with localization to~$\Omega$.

To connect this principle to general relativity, we apply it to geodesic balls on Cauchy slices.  Let $(\Sigma,h)$ be a spacelike hypersurface in a curved spacetime and $B_\Sigma(p,r)$ the $h$-geodesic ball of proper radius $r$ around a point $p\in\Sigma$.  For states that are strictly localized to $B_\Sigma(p,r)$ with Dirichlet data, we prove the intrinsic lower bound
\[
\sigma_p \;\geq\; \hbar \sqrt{\lambda_1(B_\Sigma;h)},
\]
where $\lambda_1(B_\Sigma;h)$ is the first Dirichlet eigenvalue of the Laplace--Beltrami operator $-\Delta_h$ on the ball, and $\sigma_p$ is the corresponding geometric momentum standard deviation (Theorem~\ref{thm:main}).  The estimate is \emph{intrinsic}: it depends only on the induced Riemannian metric $h$ on the slice and is manifestly insensitive to the lapse, the shift, and the extrinsic curvature.  Passing to a different slice (at a different time or in a different foliation) generally changes the geometry of $B_\Sigma(p,r)$ and hence the numerical value of $\lambda_1$, but the geometric form of the bound is the same on each slice.

The variance decomposition derived in Section~\ref{sec:var} makes this connection between geometry and uncertainty more transparent.  Writing $\psi = |\psi| e^{i\varphi}$ one obtains a kinetic-energy window
\[
\hbar\|\nabla^{h}|\psi|\|_{L^{2}} \;\le\; \sigma_{p} \;\le\; \hbar\|\nabla^{h}\psi\|_{L^{2}},
\]
which separates the contribution coming from the modulus $|\psi|$ from the fluctuations of the phase gradient in an orthonormal frame.  Operationally, this provides a clean bridge between spectral geometry on $(\Sigma,h)$ and the momentum uncertainty defined by the canonical commutator algebra on the slice.  In contrast to GUP models that postulate modified commutation relations, our construction stays within standard quantum mechanics on a fixed curved background; all “new physics’’ is carried by the geometry of the localization region.

In many situations, however, the detailed interior geometry of $B_\Sigma(p,r)$ is not known, while some information about its boundary is available.  This includes, for instance, numerical relativity setups where one can control the proper size and mean curvature of an outer boundary, or gravitational settings where a marginally outer trapped surface provides a distinguished two-sphere but the interior metric remains uncertain.  Under a weak mean-convexity assumption on the boundary (made precise in Assumption~\ref{ass:min}), the distance to $\partial B_\Sigma(p,r)$ is distributionally superharmonic.  In that case the sharp boundary-distance Hardy inequality \cite{LewisLiLi2011,DAmbrosioDipierro2014,BarbatisReview} applies and yields a purely geometric lower bound on the first Dirichlet eigenvalue in terms of the radius $r$.  Combined with our intrinsic estimate, this leads to a universal Heisenberg-type product inequality
\[
\sigma_p \,r \;\geq\; \frac{\hbar}{2}
\]
for all strictly localized states (Corollary~\ref{cor:Hardy}).  This bound is never attained, which makes it a robust, geometry-free floor on $\sigma_p$ that depends only on the proper radius $r$ and the boundary, but not on the detailed curvature in the interior.

The same weak mean-convexity hypothesis allows one to go further.  Using a vector-field version of Barta's method \cite{BessaMontenegro}, adapted to the radial distance on $B_\Sigma(p,r)$, we construct a drift field that probes only the sign of the radial Laplacian.  This yields a Barta-type lower bound on the first Dirichlet eigenvalue and hence a sharper, still scale-invariant product inequality
\[
\sigma_p\,r \;\geq\; \frac{\pi\hbar}{2}
\]
for all strictly localized states (Corollary~\ref{cor:barta}).  The improvement by a factor of $\pi$ again involves no further assumption on the slice: it is coordinate- and foliation-independent and does not require symmetry, stationarity, or vacuum conditions.  In a suitable minimax sense in three dimensions, this constant is best possible under the sole hypothesis of weak mean-convexity; any further sharpening must draw on more detailed curvature information.

Whenever such curvature information \emph{is} available, classical comparison techniques from spectral geometry can be imported directly to control $\lambda_1$ and thereby $\sigma_p$.  Cheng-type eigenvalue comparison theorems \cite{Cheng1975,BessaMontenegro2005,Yang1999,Ling2006} relate the first Dirichlet eigenvalue on a geodesic ball in $(\Sigma,h)$ to that on model spaces with prescribed Ricci curvature; small-ball asymptotics \cite{KarpPinsky1986} give precise expansions in the regime where the ball radius is much smaller than the curvature radius; and explicit radial analyses on spherically symmetric manifolds \cite{BorisovFreitas2017} provide detailed control in highly symmetric situations.  For manifolds of constant curvature, the universal product bounds obtained here are complemented by sharper curvature-dependent estimates previously derived in \cite{Schuermann2018}.

\medskip
\noindent\textbf{Organization.}
Section~\ref{sec:MinGeo} introduces the minimal geometric setting and the notion of strict Dirichlet localization to a geodesic ball; in particular, Assumption~\ref{ass:min} records the standing hypotheses on $B_\Sigma(p,r)$.  Section~\ref{sec:main} establishes the intrinsic Heisenberg-type lower bound in terms of the first Dirichlet eigenvalue (Theorem~\ref{thm:main}).  Subsections~\ref{subsec:Hardy} and~\ref{subsec:Barta} then derive the universal Hardy--Heisenberg baseline and the Barta-type improvement, respectively.  Section~\ref{sec:Examples} presents examples and applications, and Section~\ref{sec:Diss} closes with a discussion of scope and possible extensions.

\section{Minimal geometric setting}\label{sec:MinGeo}
	
	Let $(M,g)$ be a time-oriented $C^2$ Lorentzian manifold solving the Einstein field
	equations (with arbitrary stress-energy tensor and possibly $\Lambda$), and let
	$\Sigma\subset M$ be a $C^2$ spacelike hypersurface with induced Riemannian metric $h$
	and volume element $d\mu_h$.  We fix $p\in\Sigma$ and $r>0$.  Since $\Sigma$ is spacelike,
	it carries a unique future-directed unit normal field $n$ along $\Sigma$ with $g(n,n)=-1$,
	and the Levi--Civita connection $\nabla^h$ on $(\Sigma,h)$ is the tangential projection
	of the spacetime connection~$\nabla$.
	In what follows only the intrinsic Riemannian data $(\Sigma,h)$ will enter.  None of the
	extrinsic quantities (lapse, shift, or extrinsic curvature) appear in the estimates, so
	all bounds below are coordinate- and foliation-independent.

	\paragraph{Geodesic balls on the slice.}
	
	All domains are taken within $(\Sigma,h)$.  We write $B_\Sigma(p,r)$ for the $h$-geodesic ball of radius $r$ centered at $p$. Recall that the injectivity radius of $(\Sigma,h)$ at $p$, denoted $\operatorname{inj}_\Sigma(p)$, is the supremum of $r>0$ such that the Riemannian exponential map
	\[
	\exp^h_p \colon B_{T_p\Sigma}(0,r) \to \Sigma
	\]
	is a diffeomorphism onto its image; equivalently, every $h$-geodesic starting at $p$ is
	minimizing up to length~$r$.
	Throughout we choose $r$ strictly below $\operatorname{inj}_\Sigma(p)$ so that
	$B_\Sigma(p,r)$ is a bounded (Lipschitz) domain; in particular, the Dirichlet Laplacian
	on $B_\Sigma(p,r)$ has compact resolvent.  These minimal hypotheses are recorded in the following Assumption~\ref{ass:min} and will be in force for the main intrinsic estimate.

	\begin{assumption}[Minimal hypotheses]\label{ass:min}
	We assume:
	\begin{enumerate}
		\item[(i)]$r$ is strictly less than the injectivity radius $\operatorname{inj}_\Sigma(p)$ and
		chosen so that $B_\Sigma(p,r)$ is a bounded Lipschitz domain; in particular, the
		Dirichlet Laplacian has compact resolvent.

		\item[(ii)] States are strictly localized to $B_\Sigma(p,r)$: $\psi\in H^1_0(B_\Sigma(p,r))$
		with $\|\psi\|_{L^2}=1$.
	\end{enumerate}
	\end{assumption}
	\medskip\noindent
		No further assumption (e.g., symmetry, stationarity, vacuum, or asymptotic structure) is required for the main intrinsic estimate. 	
		For the Hardy and Barta baselines in Corollary \ref{cor:Hardy}  and \ref{cor:barta}, we additionally assume
		that the boundary $\partial B_\Sigma(p, r)$ of the ball is weakly mean-convex.
		For $C^2$ boundaries this is known to be equivalent to the boundary-distance
		function $d(x) = \operatorname{dist}_h(x,\partial B_\Sigma(p,r))$ being
		distributionally superharmonic on $B_\Sigma(p,r)$ (i.e., $-\Delta_h d \ge 0$);
		see Definition~\ref{def:weakly} below and \cite[Thm.\,1.2]{LewisLiLi2011} together with \cite{DAmbrosioDipierro2014}.
\medskip	
	\begin{remark} Strict localization and Dirichlet data.
		The Dirichlet condition in Assumption~\ref{ass:min} should be viewed as an idealized
		notion of strict localization or von\,Neumann--L\"uders preparation to $B_\Sigma(p,r)$. Physically, wave packets will typically
		have small ``tails'' that extend beyond any given geodesic ball. The present bounds are
		nevertheless relevant in that they describe the limiting behaviour of momentum uncertainty
		for families of states whose $L^2$-mass becomes more and more concentrated in
		$B_\Sigma(p,r)$. From a mathematical viewpoint, Dirichlet data are natural because they
		yield a purely bulk representation of the variance in terms of the Dirichlet form of
		$-\Delta_h$ and a discrete Dirichlet spectrum; for Neumann or Robin boundary conditions,
		boundary terms survive and the clean spectral control used in Theorem~\ref{thm:main}
		breaks down.
	\end{remark}
	
	\subsection{Hilbert space on geodesic balls}\label{sub:geoball}

	In the intrinsic, slice-based formulation adopted here, geodesic balls play a distinguished role as localization regions. On a spacelike hypersurface $(\Sigma,h)$, a geodesic ball $B_\Sigma(p,r)$ is defined purely in terms of the Riemannian distance induced by $h$; it does not depend on any particular coordinate chart or foliation, and it automatically adapts to the local curvature of the slice. This makes geodesic balls the most natural choice for implementing sharp, covariant position measurements in curved spacetime, in contrast to coordinate boxes or ``cubes'', whose geometry is tied to a specific parametrization of $\Sigma$.

	From the spectral point of view, the relevance of geodesic balls is even more pronounced. As shown later in Theorem~\ref{thm:main}, the intrinsic lower bound on the canonical momentum uncertainty of a state that is strictly localized in a bounded region $\Omega \subset \Sigma$ is governed by the first Dirichlet eigenvalue of the Laplace--Beltrami operator on $\Omega$. In many geometric settings, Faber--Krahn-type comparison results imply that, among all domains of a given volume, geodesic balls minimize this first Dirichlet eigenvalue. Equivalently, for fixed volume, geodesic balls yield the smallest possible spectral cost of strict localization and therefore the weakest (i.e.\ most forgiving) intrinsic lower bound on the momentum uncertainty. Any other region with the same volume has a larger first Dirichlet eigenvalue and thus enforces a stronger Heisenberg-type constraint.

	For our purposes, this means that geodesic balls provide a canonical ``worst-case'' reference geometry: if one can establish universal lower bounds for the momentum uncertainty on $B_\Sigma(p,r)$, then the corresponding bounds for other shapes of the same volume are automatically no better and typically stricter. At the same time, the round geometry of geodesic balls keeps the subsequent analysis technically manageable and well aligned with classical tools from spectral geometry. In the remainder of this subsection we therefore formulate the Hilbert space setting directly on $L^2\bigl(B_\Sigma(p,r),d\mu_h\bigr)$, taking strict Dirichlet localization to the geodesic ball as our basic notion of position uncertainty.
	
	We work on an $n$-dimensional (or $n{=}3$) Riemannian manifold $(\Sigma,h)$ and restrict attention to the geodesic ball $B_\Sigma(p,r)\subset(\Sigma,h)$. Let $\{\partial_i\}$ be a local coordinate basis and let $\{X_a\}_{a=1}^n$ be an $h$-orthonormal frame. Indices $i,j,k,\dots$ refer to coordinates (raised/lowered with $h$), while $a,b,c,\dots$ refer to the orthonormal frame (raised/lowered with $\delta$). The frames are related by a vielbein $e_a{}^{\,i}$ and its inverse $e^a{}_{\,i}$:
	\begin{equation*}
		X_a=e_a{}^{\,i}\,\partial_i,\qquad
		\partial_i=e_i{}^{\,a}\,X_a,\qquad
		e_i{}^{\,a}e_a{}^{\,j}=\delta_i^{\,j},\qquad
		e_a{}^{\,i}e_i{}^{\,b}=\delta_a^{\,b},
	\end{equation*}
	so that
	\begin{equation*}
		h_{ij}=e_i{}^{\,a}e_j{}^{\,b}\,\delta_{ab},\qquad
		h^{ij}=e^i{}_{\,a}e^j{}_{\,b}\,\delta^{ab},\qquad
		d\mu_h=\sqrt{|h|}\,d^n x .
	\end{equation*}
	All integrals and $L^2$ norms below are taken over $B_\Sigma(p,r)$ with respect to $d\mu_h$.
	For $v{=}\sum_{a=1}^n v^a X_a(x)\in T_x\Sigma$ we set
	\begin{equation*}
		\|v(x)\|_h^2 := h(v,v) = \sum_{a=1}^n |v^a(x)|^2 .
	\end{equation*}
	The pointwise gradient norm is
	\begin{equation}\label{eq:gradnorm}
		\|\nabla^h \psi(x)\|_h^2 := h(\nabla^h \psi,\nabla^h \psi)
		= \sum_{a=1}^n |X_a \psi(x)|^2 ,\qquad \psi:\Sigma\to\mathbb{C}.
	\end{equation}
	The $L^2$ inner product and norms are
	\begin{equation*}
		\langle f,g\rangle_{L^2} := \int_{B_\Sigma(p,r)} \overline{f}\,g\,d\mu_h,\qquad
		\|f\|_{L^2} := \Big(\int_{B_\Sigma(p,r)} |f|^2\,d\mu_h\Big)^{1/2},
	\end{equation*}
	and, for vector fields $Y(x)$ on $B_\Sigma(p,r)$,
	\begin{eqnarray*}
		\|Y\|_{L^2} &:=& \Big(\int_{B_\Sigma(p,r)} \|Y(x)\|_h^2\,d\mu_h\Big)^{1/2},\\
		\|\nabla^h \psi\|_{L^2} &:=& \Big(\int_{B_\Sigma(p,r)} \|\nabla^h \psi\|_h^{2}\,d\mu_h\Big)^{1/2}
		= \Big(\sum_{a=1}^n \!\int_{B_\Sigma(p,r)} |X_a\psi|^2\,d\mu_h\Big)^{\!1/2}.
	\end{eqnarray*}
\noindent
\textbf{Born probability measure.} The Hilbert space on the geodesic ball is $L^2(B_\Sigma(p,r),d\mu_h)$, where
$d\mu_h$ is the Riemannian volume element of the induced metric $h$. For a normalized state $\psi$, the density $|\psi|^2 d\mu_h$ is therefore the natural Born probability measure
for position on $B_\Sigma(p,r)$, and all expectations are taken with respect to this measure. 

\subsection{Momentum operator and variance decomposition}\label{sec:var}

Before strictly choosing the momentum operators on the curved manifold, we briefly
motivate the choice of the reference frame. While the traditional DeWitt quantization
\cite{DeWitt1957,DeWitt1965} can be consistently formulated in a natural coordinate basis
$\{\partial_i\}$, the resulting expression for the momentum variance inevitably involves the
inverse metric $h^{ij}(x)$ to contract the operator indices. This explicit metric dependence
obscures the physical interpretation, as it mixes kinematic fluctuations with geometric terms
in a non-trivial way. In contrast, adopting an $h$-orthonormal frame $\{X_a\}$ effectively
diagonalizes the local geometry, allowing the variance to be expressed as a simple sum of
squares analogous to the Euclidean case. This generalized approach is advantageous as it
facilitates a clean, coordinate-independent separation of the pure quantum fluctuations from
the intrinsic curvature effects encoded in the geometric potential derived below.\medskip

From the structural point of view, the explicit form of the momentum operators associated
with a given orthonormal frame is strongly constrained by general axiomatic schemes for
quantization on configuration manifolds. In Segal’s representation-theoretic approach to
quantization of nonlinear systems \cite{S60} and in the framework of geometric quantization
\cite{S80}, one requires that configuration observables $f\in C^\infty_0(\Sigma)$ act as
multiplication operators and that each momentum observable $P_X$ associated with a real
vector field $X$ be represented by a first-order differential operator whose principal symbol
coincides with the classical momentum. In addition, the commutator with configuration
observables is required to reproduce the classical Poisson brackets,
\begin{equation}\label{eq:CommxP}
[f, P_X] = i\hbar\,X(f),
\qquad f\in C^\infty_0(B_\Sigma(p,r)),
\end{equation}
and the resulting operators must be symmetric on $L^2(B_\Sigma(p,r),d\mu_h)$ with respect
to the Riemannian volume $d\mu_h$. The Borel quantization programme developed in
\cite{D96,D01} shows that, under these assumptions and in the presence of a fixed Riemannian
structure $(\Sigma,h)$, there is a distinguished choice for $P_X$ that is covariant under
changes of orthonormal frame and compatible with the unitary implementation of the
configuration space symmetries.

Specializing to an $h$-orthonormal frame $\{X_a\}_{a=1}^n$, the corresponding momentum
components are then realized as the differential operators
\begin{equation}
	P_a := - i\hbar\left(
	\nabla^h_{X_a} + \frac{1}{2}\,\mathrm{div}_h X_a
	\right),
	\qquad a = 1,\dots,n,
	\label{eq:Pa-DeWitt}
\end{equation}
which coincide with DeWitt’s covariant momenta in an orthonormal frame
\cite{DeWitt1957,DeWitt1965} and reduce to the familiar $p_i = - i\hbar\,\partial_i$ in
Euclidean space with a global Cartesian frame. In the language of \cite{D96,D01}, the
$\frac{1}{2}\,\mathrm{div}_h X_a$ term is precisely the correction dictated by the requirement
that the map $X\mapsto P_X$ respect the Lie algebra structure of real vector fields and that
$\sum_a P_a^2$ reproduce, up to a possible scalar curvature term, the Laplace--Beltrami
operator on $(\Sigma,h)$. With this choice, the $P_a$ generate the infinitesimal unitary action
of the isometry group on $L^2(B_\Sigma(p,r),d\mu_h)$ and provide the natural, frame-based
generalization of the canonical momenta used in flat-space quantum mechanics.

Viewed as operators on the Hilbert space $L^2(\Omega,d\mu_h)$, the $P_a$ in
\eqref{eq:Pa-DeWitt} are symmetric on the natural Dirichlet form domain (for instance
$H^1_0(\Omega)$ on a domain $\Omega \subset \Sigma$): for all $\varphi,\psi$ in that domain
one has
$\langle \varphi, P_a \psi \rangle_{L^2} {=} \langle P_a \varphi, \psi \rangle_{L^2}$.
    For a normalized state $\psi$ with $\|\psi\|_{L^2} {=} 1$ we write
    $\langle A\rangle {:=} \langle \psi, A\psi\rangle_{L^2}$ and collect the momentum
	components into the vector $P {:=} (P_a)_{a=1}^n$ with mean
	$\langle P\rangle {:=} (\langle P_1\rangle,\dots,\langle P_n\rangle)$. 
	In the present geometric setting we measure the distance using the pointwise $h$-norm,
	and for the canonical momenta $P_a$ we introduce the naive momentum
	variance of $\psi$ by
	\begin{equation}
		\mathrm{Var}^{\mathrm{naive}}_\psi(P) :=
		\big\langle (P - \langle P\rangle)^2 \big\rangle
		\equiv \sum_{a=1}^n \left( \langle P_a^2 \rangle - \langle P_a \rangle^2 \right).
		\label{eq:VarNaiveDef}
	\end{equation}
	Motivated by the structure of the quadratic form of the momenta on $H^1_0(B_\Sigma(p,r))$ (see the integration-by-parts computation in Appendix~\ref{app:DeWitt-IBP}) we have 
	\begin{equation}\label{eq:quadratic-form-identity}
		\langle P^2 \rangle = \hbar^2\|\nabla^h\psi\|_{L^2}^2
		- \hbar^2\, \langle V\rangle
	\end{equation}
	with
	\begin{equation}
		V(x) := \frac{1}{2}\sum_{a=1}^n X_a\big(f_a\big)(x)
		+ \frac{1}{4}\sum_{a=1}^n f_a(x)^2,
		\qquad x \in B_\Sigma(p,r).
		\label{eq:Vpotential}
	\end{equation}
	and $f_a := \mathrm{div}_h X_a$ for the divergence of the frame fields. 
	By construction, this potential captures exactly the curvature- and frame-dependent
	deviation of the geometric kinetic energy from the scalar Dirichlet energy. 	
	It is convenient to isolate the purely geometric contribution to the naive variance \eqref{eq:VarNaiveDef} in the form of a scalar potential \cite{DeWitt1957,DeWitt1965}. 

		\begin{definition} (Geometric momentum variance)
		Given a normalized state $\psi\in H^1_0(B_\Sigma(p,r))$, we define its
		\emph{geometric momentum variance} by
		\begin{equation}
			\mathrm{Var}_\psi(P)
			:= \big\langle (P - \langle P\rangle)^2 \big\rangle
			+ \hbar^2\langle V\rangle,
			\label{eq:VarGeom}
		\end{equation}
		that is, by adding to the naive variance \eqref{eq:VarNaiveDef} the geometric
		correction term determined by the potential $V$ in \eqref{eq:Vpotential}.
	\end{definition}

\begin{remark}[Operational distinction]
	It is important to distinguish the physical roles of the two variances defined above.
	The \emph{naive variance} $\mathrm{Var}^{\mathrm{naive}}_\psi(P)$ in \eqref{eq:VarNaiveDef} corresponds to the raw statistical spread one would obtain by measuring the frame-dependent operators $P_a$. As shown in \eqref{eq:quadratic-form-identity}, this quantity mixes the kinematic fluctuations of the state with the local curvature energy encoded in $\langle V\rangle$.
	By contrast, the \emph{geometric variance} defined in \eqref{eq:VarGeom} effectively filters out this frame-dependent potential. It isolates the intrinsic kinetic energy contribution associated with the scalar Laplace--Beltrami operator.
	Consequently, the lower bound derived in Theorem~\ref{thm:main} constrains the \emph{kinetic capacity} required for strict localization within the geometry, invariant under the choice of the orthonormal frame.
\end{remark}

	In order to relate the geometric momentum variance to the quadratic form of the
	momenta, we first rewrite the naive variance \eqref{eq:VarNaiveDef} in terms of the kinetic-energy form of $P$,
	\begin{equation}
		\mathrm{Var}^{\mathrm{naive}}_\psi(P) 
		= \langle P^2\rangle - \|\langle P\rangle\|_h^2\, .
		\label{eq:VarNaiveExpansion}
	\end{equation}
	Inserting the quadratic-form identity \eqref{eq:quadratic-form-identity} into \eqref{eq:VarNaiveExpansion} with \eqref{eq:VarGeom} this yields the intrinsic variance decomposition
	\begin{equation}
		\mathrm{Var}_\psi(P)
		= \hbar^2\|\nabla^h\psi\|_{L^2}^2 - \|\langle P\rangle\|_h^2.
		\label{eq:VarDecomposition}
	\end{equation}
	In what follows we measure momentum uncertainty by the corresponding standard deviation
	\begin{equation}
		\sigma_p(\psi) := \sqrt{\mathrm{Var}_\psi(P)}.
		\label{eq:SigmaPDef}
	\end{equation}
	
	\begin{remark}\label{rem:smallball} Physical interpretation and small-ball limit. 
		To clarify the physical significance of the potential $V$, consider its asymptotic behavior on small geodesic balls $B_\Sigma(p,r)$. In Riemann normal coordinates centered at $p$, the metric is Euclidean to first order, but the volume element expands as $\mathrm{d}\mu_h = [1 - \frac{1}{6} R_{ij}(p) x^i x^j + O(r^3)]\,\mathrm{d}^n x$. The divergence of the frame fields is determined by the logarithmic derivative of the volume density, yielding $f_a \approx \partial_a \ln \sqrt{|h|} \approx -\frac{1}{3} R_{ab}(p) x^b$. Substituting this into definition \eqref{eq:Vpotential}, the quadratic term $f_a^2$ vanishes to leading order, while the derivative term contributes $\frac{1}{2}\sum_a \partial_a f_a \approx -\frac{1}{6} \sum_a R_{aa}(p)$. Thus, for states highly localized near $p$, the expectation value behaves asymptotically as
		\begin{equation}
			\langle V \rangle \;\sim\; -\frac{1}{6}\,\mathrm{Scal}(p).
		\end{equation}
		This identifies the geometric potential $V$ as the intrinsic curvature correction -- specifically the scalar curvature term $\frac{1}{6}R$ -- that naturally arises in the quantization of systems on curved manifolds to maintain general covariance \cite{DeWitt1957}. Operationally, $\langle V \rangle$ measures the energy shift between the square of the Hermitian momentum operators and the covariant Laplace--Beltrami operator.
	\end{remark}
	
	\section{An intrinsic Heisenberg-type lower bound}\label{sec:main}
		
	From a physical standpoint, Dirichlet data $\psi|_{\partial B_\Sigma(p,r)} = 0$ idealize a
	state $\psi$ that is strictly confined to the ball-region $B_\Sigma(p,r)$, for instance in an infinite
	potential well or as the post-measurement state after a sharp localization to that ball: the
	probability of finding the particle on the boundary or outside $B_\Sigma(p,r)$ is zero.
	Neumann conditions $\partial_n \psi = 0$ instead describe perfectly reflecting boundaries with
	vanishing normal probability flux, while Robin conditions $(\partial_n \psi + \kappa \psi)=0$
	can model partially reflecting surfaces or boundary interactions localized at $\partial
	B_\Sigma(p,r)$. The intrinsic lower bound proved below thus pertains specifically to this
	idealized regime of strict confinement modeled by Dirichlet data; more general boundary
	conditions would require incorporating the corresponding surface terms into the uncertainty
	estimate and lead to a different physical problem.

Accordingly, in the following we impose Dirichlet boundary conditions throughout and work
with the Dirichlet Laplacian on bounded domains. By the Dirichlet Rayleigh--Ritz principle
(see \cite{Chavel1984}, Ch.\,3), the first Dirichlet eigenvalue $\lambda_1(\Omega; h)$ of $-\Delta_h$ on a
bounded domain $\Omega \subset \Sigma$ is characterized by
	\begin{equation}\label{eq:lambda1}
		\lambda_{1}(\Omega;h)\;=\;\inf\Bigg\{\frac{\int_{\Omega}\|\nabla^{h}u\|^{2}\,\mathrm d\mu_{h}}{\int_{\Omega}|u|^{2}\,\mathrm d\mu_{h}}: \; 0\neq u\in H^{1}_{0}(\Omega)\Bigg\}.
	\end{equation}
	
	\begin{theorem}[Intrinsic Heisenberg-type lower bound]\label{thm:main}
		Let Assumption~\ref{ass:min} hold. Then, for every normalized $\psi\in H^{1}_{0}(B_{\Sigma}(p,r))$,
		\begin{equation}\label{eq:IHB}
			\sigma_{p}(\psi)\;\ge\;\hbar\,\sqrt{\lambda_{1}\!\left(B_{\Sigma};h\right)}.
		\end{equation}
	Equality holds if and only if $|\psi|$ is a first Dirichlet eigenfunction of $-\Delta_h$
	on $B_\Sigma(p, r)$ and the phase gradient $\nabla_h(\arg \psi)$ is constant almost
	everywhere on $B_\Sigma(p, r)$ with respect to the probability measure $|\psi|^2\,d\mu_h$.
	\end{theorem}
	\noindent
	\textbf{Proof.} 
	Write $\psi = u\,e^{i\phi}$ with $u = |\psi|\ge 0$, and fix an $h$-orthonormal frame $\{X_a\}_{a=1}^n$ on $B_\Sigma(p,r)$. For Dirichlet data ($u|_{\partial B_\Sigma(p,r)}=0$), integration by parts yields the components of the mean momentum
	\begin{equation}
		\langle P_a\rangle
		= \hbar \int_{B_\Sigma(p,r)} u^{2}\, X_a \phi \, d\mu_h,
		\qquad a=1,\dots,n .
	\end{equation}
	Since $\psi = u\, e^{i\phi}$ with $u = |\psi| \ge 0$ and real-valued phase $\phi$,
	differentiate in an $h$-orthonormal frame $\{X_a\}_{a=1}^n$ to obtain
	\[
	X_a\psi = e^{i\phi}(X_a u + i u\,X_a\phi), \qquad a=1,\dots,n.
	\]
	Hence
	\[
	|X_a\psi|^2 = |X_a u|^2 + u^2 |X_a\phi|^2
	\]
	and, using the pointwise norm \eqref{eq:gradnorm},
	\[
	\|\nabla_h\psi\|_h^2=
		\sum_{a=1}^n |X_a\psi|^2 = \sum_{a=1}^n \bigl(|X_a u|^2 + u^2 |X_a\phi|^2\bigr)
			= \|\nabla_h u\|_h^2 + u^2 \|\nabla_h\phi\|_h^2.
	\]
	This pointwise decomposition is the (local) Madelung identity. It holds first for smooth
	$\psi$, and by density extends to $\psi \in H^1_0(B_\Sigma(p,r))$.	 	
	Inserting it into the variance formula \eqref{eq:VarDecomposition} yields the exact decomposition
	\begin{equation}\label{eq:VarX}
		\mathrm{Var}_{\psi}(P)
		= \hbar^{2}\!\int_{B_\Sigma(p,r)}
		\big(\|\nabla^{h}u\|_h^{2} + u^{2}\|\nabla^{h}\phi\|_h^{2}\big)\,d\mu_h
		\;-\; \hbar^{2}\sum_{a=1}^{n}
		\Big(\,\int_{B_\Sigma(p,r)} u^{2}\, X_a \phi \, d\mu_h\Big)^{\!2}.
	\end{equation}
	Since $\|\psi\|_{L^2(B_\Sigma(p,r),d\mu_h)} {=} 1$ and $u {=} |\psi|$, the weight
	$u^2\,d\mu_h {=} |\psi|^2\,d\mu_h$ corresponds to the quantum mechanical Born probability measure on $B_\Sigma(p,r)$ introduced at the end of Sec.\,\ref{sub:geoball}.
	Applying the Cauchy--Schwarz inequality with respect to this probability measure, we obtain
	\begin{equation}
		\sum_{a=1}^n \Bigl( \int_{B_\Sigma(p,r)} u^2 X_a\phi \,d\mu_h \Bigr)^2
		\;\le\; \int_{B_\Sigma(p,r)} u^2 \|\nabla_h\phi\|_h^2 \,d\mu_h.
	\end{equation}
	Thus the last term in \eqref{eq:VarX} is bounded above by
	$\hbar^2 \int u^2 \|\nabla_h\phi\|_h^2\,d\mu_h$, and inserting this estimate into
	\eqref{eq:VarX} yields the lower bound
	\begin{equation}
		\mathrm{Var}_\psi(P)
		\;\ge\; \hbar^2 \int_{B_\Sigma(p,r)} \|\nabla_h u\|_h^2 \,d\mu_h
	\end{equation}
	With $u=|\psi|$ this yields
	\begin{equation}\label{eq:VarBound}
		\mathrm{Var}_{\psi}(P) \;\ge\; \hbar^{2}\,\|\nabla^{h}|\psi|\,\|_{L^{2}}^{2}.
	\end{equation}
	Setting $d\nu := u^{2}\,d\mu_h$ and, for any scalar field $X$,
	\begin{equation}
		E_{\nu}[X] := \int_{B_\Sigma(p,r)} X \, d\nu ,
	\end{equation}
	we rewrite \eqref{eq:VarX} in terms of $\nu$-expectations as
	\begin{equation}\label{eq:diff0}
		\mathrm{Var}_\psi(P)
		= \hbar^2 \|\nabla^h u\|_{L^2}^2
		+ \hbar^2\Big( E_\nu[\|\nabla^h\phi\|_h^2] - \sum_{a=1}^n (E_\nu[X_a\phi])^2 \Big).
	\end{equation}
	Here and in the sequel, “a.e.” (“almost everywhere”) is always understood with respect to
	the relevant underlying measure; in particular, “$\nu$-a.e.” refers to the probability
	measure $d\nu = u^2\,d\mu_h$ on $B_\Sigma(p,r)$.
	Invoking \eqref{eq:gradnorm}, the difference of expectations simplifies to
	\begin{eqnarray}\label{eq:diff}
		E_\nu[\|\nabla^h\phi\|_h^2] - \sum_{a=1}^n (E_\nu[X_a\phi])^2
		=  \sum_{a=1}^n E_\nu\!\left[\, \big(X_a\phi -E_\nu[X_a\phi]\big)^{\!2}\,  \right].
	\end{eqnarray}
	In \eqref{eq:diff} the fluctuation term is written as a sum of variances of the scalar fields
	$X_a\phi$ with respect to the probability measure $d\nu$. One may rewrite this
	sum as the squared $h$-norm of a \emph{centered} gradient vector and then take the $\nu$-expectation. Concretely,
	\[
	\sum_{a=1}^n E_\nu\Big[\big(X_a\phi - E_\nu[X_a\phi]\big)^{\!2}\Big]
	= E_\nu\!\left[\sum_{a=1}^n \big(X_a\phi - E_\nu[X_a\phi]\big)^{\!2}\right]
	= E_\nu\!\left[\;\|\,\nabla^h\phi - E_\nu[\nabla^h\phi]\,\|_h^{2}\,\right].
	\]
	The first equality is linearity of $E_\nu[\,\cdot\,]$. The second uses the orthonormality of the frame $\{X_a\}_{a=1}^n$ together with the definition of the pointwise $h$-norm in \eqref{eq:gradnorm}. With \eqref{eq:diff0} and \eqref{eq:diff} the variance admits the exact decomposition
	\begin{equation}\label{eq:VarE2}
		\mathrm{Var}_\psi(P)
		= \hbar^2 \|\nabla^h u\|_{L^2}^2
		+ \hbar^2 E_\nu\!\left[\,\|\,\nabla^h \phi - E_\nu[\nabla^h \phi]\,\|_h^{2}\,\right],
	\end{equation}
	for $\psi= u\, e^{i\phi}$. Hence, equality in \eqref{eq:VarBound} holds precisely when the gradient $\nabla^h\phi$ is $\nu$-a.e.\ constant. This is the \emph{phase} part of the equality condition in Theorem\,\ref{thm:main}; together with $|\psi|$ being a first Dirichlet eigenfunction it yields equality in \eqref{eq:IHB}.
	By the Dirichlet Rayleigh--Ritz characterization \eqref{eq:lambda1} (see \cite[Ch.~3]{Chavel1984}), for any $0\neq u\in H^{1}_{0}(B_{\Sigma}(p,r))$,
	\begin{equation}\label{eq:Rayleigh}
		\int_{B_{\Sigma}(p,r)}\|\nabla^{h}u\|_h^{2}\,\mathrm d\mu_{h}
		\;\ge\; \lambda_{1}\!\left(B_{\Sigma}(p,r);h\right)\int_{B_{\Sigma}(p,r)}|u|^{2}\,\mathrm d\mu_{h}.
	\end{equation}
	Taking $u = |\psi|$ with $\|\psi\|_{L^2} = 1$ (hence $\|u\|_{L^2} = 1$) and combining
	\eqref{eq:VarBound} with \eqref{eq:Rayleigh}, we obtain the chain
	\begin{equation}
	\sigma_{p}(\psi)^{2}
	= \mathrm{Var}_{\psi}(P)
	\;\ge\;\hbar^{2}\!\int_{B_{\Sigma}(p,r)}\|\nabla^{h}|\psi|\,\|_h^{2}\,\mathrm d\mu_{h}
	\;\ge\;\hbar^{2}\lambda_{1}\!\left(B_{\Sigma}(p,r);h\right).
	\end{equation}
	Taking square roots gives the desired estimate $\sigma_p(\psi) \ge \hbar
	\sqrt{\lambda_1}$, i.e. \eqref{eq:IHB}.\hfill$\square$
\medskip
	\begin{remark}\label{rem:kinetic-window}
	Coordinate invariance and kinetic-energy window.
	The estimate is purely intrinsic to the Riemannian data $(B_\Sigma(p, r), h)$:
	neither coordinates nor the ambient foliation play any role in its formulation or
	proof.
	The kinetic-energy window follows directly from \eqref{eq:VarDecomposition} and the exact decomposition
	\eqref{eq:VarE2}. On the one hand, \eqref{eq:VarDecomposition} reads
	\[
	\mathrm{Var}_\psi(P)
	= \hbar^2 \|\nabla_h \psi\|_{L^2}^2 - \|\langle P\rangle\|_h^2
	\le \hbar^2 \|\nabla_h \psi\|_{L^2}^2,
	\]
	since $\|\langle P\rangle\|_h^2 \ge 0$. On the other hand, the expectation in term in \eqref{eq:VarE2} is nonnegative, so
	\[
	\mathrm{Var}_\psi(P) \ge \hbar^2 \|\nabla_h u\|_{L^2}^2
	= \hbar^2 \|\nabla_h|\psi|\|_{L^2}^2.
	\]
	Taking square roots in the two-sided inequality for $\mathrm{Var}_\psi(P)$ gives
	\begin{equation}\label{eq:chain}
		\hbar\|\nabla^{h}|\psi|\|_{L^2}\;\le\;\sigma_{p}(\psi)\;\le\;\hbar\|\nabla^{h}\psi\|_{L^2}.
	\end{equation}
	which is the claimed kinetic-energy window.
	Equality at the upper endpoint holds precisely when the mean momentum vanishes,
	$\|\langle P\rangle\|_h = 0$ (equivalently, $\langle P\rangle = 0$), so that the
	subtracted term in \eqref{eq:VarDecomposition} disappears. Equality at the lower endpoint holds precisely
	when the fluctuation term in \eqref{eq:VarE2} vanishes, i.e. when $\nabla_h\phi$ is constant
	($\nu$-a.e.). If this lower-endpoint equality holds but $|\psi|$ is not a
	ground-state Dirichlet eigenfunction, then the Rayleigh quotient is strict and the
	bound \eqref{eq:IHB} is correspondingly strict: $\sigma_p(\psi) > \hbar\sqrt{\lambda_1}$.
	\end{remark}
\medskip
	\begin{remark}\label{rem:phase-const} 
		On the ‘constant phase gradient’ condition and frame independence. 
		By the exact variance decomposition \eqref{eq:VarE2},
		equality at the \emph{lower} endpoint of the kinetic-energy window \eqref{eq:chain}
		(equivalently of \eqref{eq:VarBound}) holds if and only if the $\nu$-variance of the phase gradient vanishes:
		\[
		E_\nu\!\big[\;\|\nabla^h\phi - E_\nu[\nabla^h\phi]\|_h^2\;\big]=0
		\quad\Longleftrightarrow\quad
		\nabla^h\phi(x)=E_\nu[\nabla^h\phi]\ \ \text{for $\nu$-a.e.\ $x\in B_\Sigma(p,r)$}.
		\]
		In a global $h$-orthonormal frame $\{X_a\}_{a=1}^n$ this is equivalent to the existence of constants $c_a$
		with $X_a\phi(x)=c_a$ for $\nu$-a.e.\ $x$.
		
		\medskip
		\noindent\emph{Global frame and invariance.}
		Since $r{<}\operatorname{inj}_\Sigma(p)$, the ball $B_\Sigma(p,r)$ is contractible. Parallel transport of an orthonormal basis at $p$ along radial geodesics furnishes a global $h$-orthonormal frame. The quantity $E_\nu[\|\nabla^h\phi - E_\nu[\nabla^h\phi]\|_h^2]$ and the condition above are invariant under
		\emph{constant} $O(n)$ rotations of this fixed global frame (not under $x$-dependent frame rotations).
		
		\medskip
		\noindent\emph{Realizability / integrability.}
		Writing $v:=\sum_{a=1}^n c_a\,X_a$, the condition above requires $\nabla^h\phi=v$. Such a phase exists if and only $\nabla_i v_j-\nabla_j v_i=0$, for all $i,j$.
		On curved balls this typically fails; hence equality in \eqref{eq:IHB} is often \emph{not} attained (the bound is sharp but strict). In the Euclidean ball $B_r(0)\subset\mathbb{R}^n$, however, $\nabla_i v_j=\nabla_j v_i$ for constant $v$, and affine phases $\phi(x)=k\!\cdot\!x$ realize the lower endpoint together with the Dirichlet ground-state modulus.
	\end{remark}
 	\begin{remark}\label{rem:saturation}
	Example saturating the bound with nonzero mean momentum.
		Let $B_r(0)\subset\mathbb{R}^n$ carry the Euclidean metric $h$. Let $u_1>0$ denote the $L^2$-normalized first Dirichlet eigenfunction of $-\Delta_h$ on $B_r(0)$, with eigenvalue $\lambda_1(B_r(0);h)=: \lambda_1$. For fixed $k\in\mathbb{R}^n$ set
		\[
		\psi(x):=u_1(x)\,e^{i k\cdot x},\qquad x\in B_r(0).
		\]
		Writing $\psi=u_1 e^{i\phi}$ with $|\psi|=u_1$ and $\phi(x)=k\cdot x$, we have $\nabla^{h}\phi\equiv k$, hence each component $X_a\phi$ is constant. In the exact variance decomposition \eqref{eq:VarE2}, the fluctuation term vanishes and the lower endpoint of the kinetic-energy window \eqref{eq:chain} is attained:
		\[
		\sigma_p(\psi)=\hbar\,\|\nabla^{h} u_1\|_{L^2}=\hbar\sqrt{\lambda_1}.
		\]
		At the same time the mean momentum equals $\langle P\rangle=\hbar k$, which is nonzero unless $k=0$. Thus equality in the intrinsic Heisenberg bound \eqref{eq:IHB} may hold even when the mean momentum does not vanish; it suffices that the phase gradients are (a.e.) constant together with $|\psi|$ being a first Dirichlet eigenfunction.
	\end{remark}
	
	\subsection{Universal Hardy baseline on weakly mean-convex balls}\label{subsec:Hardy}
	
	Informally, the boundary $\partial B_{\Sigma}(p,r)$ of the geodesic ball is \emph{weakly mean-convex} if it never “bulges outward” with negative mean curvature: the mean curvature with respect to the outward unit normal is nonnegative in the distributional sense, which in this paper is \emph{exactly equivalent} to the boundary-distance\footnote{Here \(\operatorname{dist}_{h}\) denotes the geodesic (Riemannian) distance on \(\Sigma\) induced by \(h\).}
	function $d(x)=\operatorname{dist}_{h}(x,\partial B_{\Sigma}(p,r))$ being superharmonic on $B_{\Sigma}(p,r)$. This is the only extra hypothesis beyond Assumption~\ref{ass:min} used to obtain the sharp boundary-distance Hardy estimate \eqref{eq:Hardy}, and -- since $d(x)\le r$ a.e. -- the universal baseline product \eqref{eq:HardyBaseline} recorded in Corollary~\ref{cor:Hardy}.
	
	\begin{definition}\label{def:weakly} [Weakly mean-convex domain]
		Let $(\Sigma,h)$ be a Riemannian manifold and $\Omega\subset\Sigma$ a bounded $C^{2}$ domain. We say that $\Omega$ is \emph{weakly mean-convex} if the mean curvature $H$ of $\partial\Omega$
		with respect to the outward unit normal is nonnegative in the distributional sense. Equivalently,
		the boundary-distance function $d(x)=\operatorname{dist}_{h}(x,\partial\Omega)$ is
		distributionally superharmonic on $\Omega$ (i.e., $-\Delta_{h}d\geq 0 $)
		-- see \cite[Thm.~1.2]{LewisLiLi2011} and \cite{DAmbrosioDipierro2014}.
	\end{definition}
	
	\begin{corollary}[Hardy baseline: optimal, not attained]\label{cor:Hardy}
		Assume, in addition to Assumption~\ref{ass:min}, that the boundary \(\partial B_{\Sigma}(p,r)\) is weakly mean-convex. Then, for every normalized \(\psi\in H^{1}_{0}\!\left(B_{\Sigma}(p,r)\right)\), the universal product bound holds:
		\begin{equation}\label{eq:HardyBaseline}
			\sigma_{p}(\psi)\,r\;\ge\;\frac{\hbar}{2}.
		\end{equation}
		The constant \(1/2\) is optimal, and the bound in \eqref{eq:HardyBaseline} is never attained.
	\end{corollary}
	
	\begin{proof}
		On weakly mean-convex $C^2$ domains the boundary-distance function is superharmonic
		in the distributional sense, and the sharp boundary-distance Hardy inequality holds: for every $u \in H^1_0(\Omega)$,		
		\[
		\int_\Omega \|\nabla^h u\|_h^2\,d\mu_h \;\ge\; \frac14
		\int_\Omega \frac{|u|^2}{d(x)^2}\,d\mu_h,
		\]
		where $d(x) := \operatorname{dist}_h(x,\partial\Omega)$ (see \cite{LewisLiLi2011,DAmbrosioDipierro2014,BarbatisReview}).
		Specialized to $\Omega = B_\Sigma(p,r)$ this gives
		\begin{equation}\label{eq:Hardy}
			\int_{B_{\Sigma}(p,r)} \|\nabla^{h}u\|^{2}_h\,\mathrm d\mu_{h}
			\;\ge\; \frac{1}{4}\int_{B_{\Sigma}(p,r)} \frac{|u|^{2}}{d(x)^{2}}\,\mathrm d\mu_{h},
		\end{equation}
		where \(d(x)=\operatorname{dist}_{h}\!\bigl(x,\partial B_{\Sigma}(p,r)\bigr)\). Since \(d(x)\le r\) for almost every \(x\), \eqref{eq:Hardy} implies
		\begin{equation}
			\int_{B_{\Sigma}(p,r)} \|\nabla^{h}u\|^{2}_h\,\mathrm d\mu_{h}
			\;\ge\; \frac{1}{4r^{2}}\int_{B_{\Sigma}(p,r)} |u|^{2}\,\mathrm d\mu_{h}.
		\end{equation}
		Taking the infimum over \(u\in H^{1}_{0}\!\left(B_{\Sigma}(p,r)\right)\setminus\{0\}\) yields
		\(\lambda_{1}\!\left(B_{\Sigma}(p,r);h\right)\ge 1/(4r^{2})\).
		Combining this with the intrinsic lower bound \eqref{eq:IHB} from Theorem~\ref{thm:main}
		produces the product estimate \eqref{eq:HardyBaseline}. Optimality of the constant and non-attainment are classical features of Hardy’s inequality and persist on mean-convex domains; see, e.g., \cite{BarbatisReview}.
	\end{proof}
	
	\paragraph{Compatibility with curvature-dependent bounds on spheres.}
	Consider a slice \((\Sigma,h)\) of constant sectional curvature \(\kappa>0\) (e.g.\ the round \(S^3\)). For geodesic balls \(B_\Sigma(p,r)\) one has weak mean-convexity for \(0<r\le \pi/(2\sqrt{\kappa})\);
	hence Corollary~\ref{cor:Hardy} applies and yields the universal baseline
	 \eqref{eq:HardyBaseline}. By contrast, the curvature-dependent estimate in \cite{Schuermann2018}, for $\kappa>0$,
	\begin{equation}\label{eq:Schuermann19}
		\sigma_p(\psi)\, r \;\ge\; \pi \hbar \,\sqrt{\,1 - \frac{\kappa\,}{\pi^{2}}\, r^{2}\,}\,,
	\end{equation}
	is valid on the larger interval \(0<r<\pi/\sqrt{\kappa}\) and decreases to \(0\) as \(r\to \pi/\sqrt{\kappa}\) (antipodal radius), where mean-convexity fails. On the overlap \(0<r\le \pi/(2\sqrt{\kappa})\) the bound \eqref{eq:Schuermann19} is strictly stronger than \eqref{eq:HardyBaseline}; for instance, at \(r=\pi/(2\sqrt{\kappa})\) it yields \(\sigma_p\,r\ge \pi\hbar\sqrt{3}/2\gg \hbar/2\).
	Thus there is no contradiction: \eqref{eq:HardyBaseline} is a foliation-independent, boundary-driven floor available under mean-convexity, while \eqref{eq:Schuermann19} is a curvature-sensitive refinement that extends beyond the convexity radius.
	
	\subsection{Barta-type improvement under weak mean-convexity}\label{subsec:Barta}
	
	The universal Hardy baseline already enforces the scale-invariant floor \eqref{eq:HardyBaseline} under the weak mean-convexity hypothesis in Assumption~\ref{ass:min}. In this subsection we sharpen that floor using a vector-field version of Barta's method \cite{BessaMontenegro}. The idea is entirely intrinsic to \((\Sigma,h)\): one tests the Dirichlet form with a radially aligned drift field adapted to the distance function, so that only the sign of the radial Laplacian enters. Combined with the intrinsic lower bound of Theorem~\ref{thm:main}, this yields a factor-\(\pi\) improvement of the product inequality, recorded as \eqref{eq:barta}. The proof is short and foliation-independent, requires no symmetry or vacuum assumptions, and rests on the Barta identity in the formulation of \cite{BessaMontenegro}. We state the result next as Corollary~\ref{cor:barta}.
	
	\begin{corollary}[Barta-type improvement under weak mean-convexity]\label{cor:barta}
		Assume that ${B=B_\Sigma(p,r)}$ is weakly mean-convex,  
		equivalently $\Delta \rho\ge 0$ on $B\!\setminus\!\{p\}$,
		where $\rho(x) := \mathrm{dist}_h(x,p)$. Then
		\begin{equation}\label{eq:barta}
			\sigma_p(\psi)\,r\ \ge\ \frac{\pi\hbar}{2}\, 
		\end{equation}
		for all $\psi\in H^1_0(B), \|\psi\|_2=1$.
	\end{corollary}
	
	\begin{proof}
		We use the vector-field version of Barta's method from \cite{BessaMontenegro}. For any $X\in C^1(B,T\Sigma)$ and any $u\in H^1_0(B)$ one has the identity
		\begin{equation}\label{eq:barta-identity}
			\int_B |\nabla u|^2
			=\int_B \bigl|\nabla u - Xu\bigr|^2 - \int_B (\mathrm{div}X+|X|^2)\,u^2,
		\end{equation}
		hence 
		\begin{equation}\label{eq:BartaVec}
			\int_B |\nabla u|^2\ \ge\ \int_B \bigl(-\mathrm{div}X-|X|^2\bigr)\,u^2
			\quad\Longrightarrow\quad
			\lambda_1(B)\ \ge\ \inf_B\bigl(-\mathrm{div}X-|X|^2\bigr).
		\end{equation}
		Choose $\alpha:=\pi/(2r)$ and the radial field
		\begin{equation*}
			X\ :=\ -\,\alpha\,\tan\bigl(\alpha\,\rho\bigr)\,\nabla \rho.
		\end{equation*}
		Using $\mathrm{div}(\phi(\rho)\nabla \rho)=\phi'(\rho)+\phi(\rho)\,\Delta \rho$ (valid on $B\setminus\{p\}$; a standard cut-off around $p$ justifies integration), and $|\nabla \rho|=1$, we compute
		\begin{equation*}
			-\mathrm{div}X-|X|^2
			=\alpha^2+\alpha\,(\Delta \rho)\,\tan(\alpha \rho).
		\end{equation*}
		Since $0\le \alpha \rho<\tfrac{\pi}{2}$ on $B$ and $\tan(\alpha \rho)\ge 0$, the weak mean-convexity hypothesis $\Delta \rho\ge 0$ yields
		\begin{equation*}
			-\mathrm{div}X-|X|^2\ \ge\ \alpha^2\ =\ \frac{\pi^2}{4r^2}.
		\end{equation*}
		Taking the infimum over $B$ gives $\lambda_1(B)\ge \pi^2/(4r^2)$. The second claim then follows from Theorem\,\ref{thm:main}.
	\end{proof}
	
	\begin{remark} Context for the bounds \eqref{eq:HardyBaseline} and \eqref{eq:barta}.
		Under Assumption~\ref{ass:min}, the identity $d=r-\rho$ reduces the boundary-distance Hardy hypothesis to the radial condition used in the Barta argument. Thus the improvement from the universal Hardy baseline \eqref{eq:HardyBaseline} to the Barta-type estimate \eqref{eq:barta} comes from the vector-field method and choice of $X$, not from a stronger geometric assumption. This improvement is scale-invariant and uses only the sign of \(\Delta \rho\), hence it is robust, geometry- and foliation-independent within the present framework. In minimax sense and three dimensions 
		(see, e.g., the discussion in \cite{BessaMontenegro}), \eqref{eq:barta} is best possible under exactly the hypotheses of Corollary~\ref{cor:barta}; any further strengthening requires additional information (e.g.\ curvature/comparison hypotheses), cf.\ \cite{BessaMontenegro,Cheng1975} or  curvature-dependent refinement \cite{Schuermann2018}.
	\end{remark}
	
	\section{Examples and applications}\label{sec:Examples}
	
	An illuminating perspective on the geometric potential $V$ arises when the Riemannian manifold $(\Sigma,h)$ is constructed from a Lie algebra. Instead of deriving the frame from a metric, one specifies the commutation relations of an orthonormal frame $\{X_a\}_{a=1}^n$,
	\begin{equation}
		[X_a, X_b] = \sum_{c=1}^n C_{ab}^c X_c,
	\end{equation}
	where $C_{ab}^c$ are the structure constants satisfying the Jacobi identity. This construction yields a Lie group (or a quotient thereof) equipped with a left-invariant metric $h(X_a, X_b)=\delta_{ab}$. In this algebraic setting, the divergence of the frame fields is constant and determined by the trace of the structure constants (adjoint representation),
	\begin{equation}
		f_a = \mathrm{div}_h X_a = - \sum_{b=1}^n C_{ab}^b .
	\end{equation}
	Since $f_a$ is constant, the derivative term in \eqref{eq:Vpotential} vanishes, and the geometric potential becomes a pure algebraic invariant:
	\begin{equation}
		V = \frac{1}{4} \sum_{a=1}^n f_a^2 \;>\;0.
	\end{equation}
	This classification creates a dichotomy between unimodular and non-unimodular geometries, which directly impacts the intrinsic momentum uncertainty.
	
	\paragraph{Example 4.1} [Category A: Unimodular geometry / Heisenberg group]\\	
	Consider the 3-dimensional Heisenberg algebra $\mathfrak{h}_3$ generated by $\{X_1, X_2, X_3\}$ with the non-vanishing commutator
	\begin{equation}
		[X_1, X_2] = X_3.
	\end{equation}
	A concrete realization of this algebra is provided by the \textit{Nil geometry}, one of the eight Thurston model geometries. On the manifold $\Sigma \cong \mathbb{R}^3$ with coordinates $(x,y,z)$, the orthonormal frame generators are explicitly given by:
	\begin{equation}
		X_1 = \partial_x - \frac{1}{2}y \partial_z, \quad
		X_2 = \partial_y + \frac{1}{2}x \partial_z, \quad
		X_3 = \partial_z.
	\end{equation}
	The Riemannian metric $h$ that renders this frame orthonormal is the left-invariant metric
	\begin{equation}
		ds^2_h = dx^2 + dy^2 + \left(dz + \frac{1}{2}(y dx - x dy)\right)^2.
	\end{equation}
	Returning to the algebraic perspective, the only non-zero structure constants are $C_{12}^3 = -C_{21}^3 = 1$. Computing the divergences $f_a = -\sum_b C_{ab}^b$ yields
	$f_1 = f_2 = f_3 = 0$.
	Consequently, the algebra is unimodular, and the geometric potential vanishes identically, $V \equiv 0$. 
	
	Although the geometric potential $V$ vanishes identically, the geometry is non-flat: the scalar curvature is constant and negative, $\mathrm{Scal} = -\tfrac{1}{2}$. Moreover, for the maximally symmetric left-invariant Riemannian Nil metric considered here, the injectivity radius is finite and spatially constant, $\operatorname{inj}_{\Sigma}(p) \equiv 2\pi$ for all $p \in \Sigma$. Consequently, the strict localization assumption (Assumption~\ref{ass:min}) imposes a specific upper bound on the size of the geodesic ball in this geometry, requiring 
	\[
		0 < r < 2\pi\, .
	\]
	The value of the injectivity radius and its point-independence follow from the explicit description of the geodesic flow, conjugate and cut loci for left-invariant Riemannian metrics on the Heisenberg groups in Biggs and Nagy~\cite{BiggsNagy2016}.

	\paragraph{Example 4.2} [Category B: 3D Hyperbolic Space $\mathbb{H}^3$]\\
	An analytically tractable example of a non-unimodular geometry is the three-dimensional hyperbolic space $\mathbb{H}^3$ with curvature $\kappa=-1$. We model $\mathbb{H}^3$ as the group manifold generated by two translations $\{X_1, X_2\}$ and one dilation $X_3$, satisfying
	\begin{equation*}
		[X_3, X_1] = X_1, \qquad [X_3, X_2] = X_2, \qquad [X_1, X_2] = 0.
	\end{equation*}
	The non-vanishing structure constants are $C_{31}^1 = C_{32}^2 = 1$. Summing over the adjoint indices yields the divergence vector of the frame:
	$f_3 = - (C_{31}^1 + C_{32}^2) = -2$ and $f_1 = f_2 = 0.$
	Inserting this into the definition of the geometric potential \eqref{eq:Vpotential} gives the constant algebraic shift \[V= \frac{1}{4} f_3^2 = 1\, .\]
	Conversely, consider the Dirichlet problem on a geodesic ball of radius $R$ in $\mathbb{H}^3$. The radial Laplacian $\Delta_r = \partial_r^2 + 2\coth(r)\partial_r$ transforms via $u(r) = \chi(r)/\sinh(r)$ into the Schrödinger operator $-\partial_r^2 + 1$. The ground state $\chi(r) = \sin(\pi r/R)$ corresponds to the eigenvalue 
	\begin{equation*}
		\lambda_1(\mathbb{H}^3, R) = \frac{\pi^2}{R^2}+1.
	\end{equation*}
	The algebraic potential $V=1$ exactly matches the spectral gap $\lim_{R\to\infty} \lambda_1 = 1$. The intrinsic Heisenberg bound thus takes the explicit form \cite{Schuermann2018}
	\begin{equation*}
		\sigma_p(\psi)^2 \;\ge\; \hbar^2 \left( \frac{\pi^2}{R^2}+1 \right),
	\end{equation*}
	demonstrating how the non-unimodular algebra enforces a strictly positive lower bound on momentum uncertainty even in the infinite-volume limit. 
	Using the variance decomposition \eqref{eq:VarGeom}, we can isolate the naive variance measured by the momentum operators in this specific frame. For the ground state, we obtain:
	\begin{align*}
		\mathrm{Var}^{\mathrm{naive}}_\psi(P) 
		&\;=\; \hbar^2 \frac{\,\pi^2}{R^2}.
	\end{align*}
	Strikingly, the naive variance recovers the exact Euclidean scaling $1/R^2$. The geometric potential $V$ acts as a filter that subtracts the spectral gap (the pure curvature energy) from the total intrinsic uncertainty, leaving behind the pure confinement energy associated with the ball size $R$.
	The injectivity radius is infinite, $\operatorname{inj}_{\Sigma} = \infty$, implying that the radius $r$ of the geodesic ball is unconstrained.

	\paragraph{Example 4.3} [Non-constant curvature: ``Witten's Cigar'']\\
	To demonstrate the scope of the intrinsic bound beyond homogeneous spaces, we consider the 2-dimensional manifold $\Sigma \cong \mathbb{R}^2$ equipped with the rotationally symmetric metric
	\begin{equation}
		h \;=\; \mathrm{d}r^2 + \tanh^2(r)\,\mathrm{d}\phi^2, \qquad r\ge 0,\ \phi\in[0,2\pi).
	\end{equation}
	This geometry, known as \emph{Witten's cigar} or the Euclidean 2D black hole, interpolates smoothly between a spherical cap of positive Gaussian curvature $K=2$ at the origin and a flat cylinder ($K\to 0$) at infinity. 
	A natural orthonormal frame $\{X_a\}$ is obtained by normalizing the coordinate basis:
	\begin{equation}
		X_1 = \partial_r, \qquad X_2 = \coth(r)\,\partial_\phi.
	\end{equation}
	Unlike in the Heisenberg or hyperbolic cases, the algebra of these fields is governed by position-dependent structure functions rather than constant coefficients. A direct computation of the Lie bracket yields
	\begin{equation}
		[X_1, X_2] \;=\; -\frac{2}{\sinh(2r)}\,X_2.
	\end{equation}
	The non-vanishing structure function $C_{12}^2(r) = -2/\sinh(2r)$ encodes the inhomogeneity of the curvature. 
	The divergences of the frame fields are $f_1 = \mathrm{div}_h X_1 = 2/\sinh(2r)$ and $f_2=0$. Consequently, the geometric potential $V$ defined in \eqref{eq:Vpotential} becomes a spatially varying field:
	\begin{equation}\label{eq:VCigar}
		V(r) \;=\; \frac{1}{2}\partial_r f_1 + \frac{1}{4}f_1^2 \;=\; \frac{1 - 2\cosh(2r)}{\sinh^2(2r)}.
	\end{equation}
	The potential \eqref{eq:VCigar} captures the global topology of the manifold:
	\begin{itemize}
		\item As $r\to \infty$, the potential decays exponentially, $V(r) \to 0$, consistent with the asymptotic flatness of the cylinder.
		\item As $r\to 0$, it diverges as $V(r) \sim -1/(4r^2)$. This singular behavior is not pathological but physical: it exactly cancels the centrifugal barrier term coming from the Laplacian in polar coordinates, ensuring that the effective Hamiltonian is regular at the origin (where the geometry closes smoothly like $\mathbb{R}^2$).
	\end{itemize}
	Finally, the first Dirichlet eigenvalue $\lambda_1$ on a ball $B_R(0)$ is determined by the zeros of the Legendre function $P_{-1/2+i\nu}(\cosh 2R)$. Theorem~\ref{thm:main} thus provides a strictly intrinsic lower bound $\sigma_p \ge \hbar\sqrt{\lambda_1}$ that automatically incorporates these non-trivial geometric features. 
	For a detailed comparison between the intrinsic lower bound derived in Theorem~\ref{thm:main} and classical spectral estimates in different radial regimes we consider the spectral asymptotics and sharpness of the bound:
	
	\begin{itemize}
		\item \textbf{Small-ball regime ($R \ll 1$).} 
		Near the origin, the geometry is dominated by the positive Gaussian curvature $K(0)=2$. The asymptotic expansion for small geodesic balls \cite{KarpPinsky1986} yields
		\begin{equation}
			\lambda_1(B_R) \;=\; \frac{j_{0,1}^2}{R^2} - \frac{K(0)}{3} + \mathcal{O}(R^2),
		\end{equation}
		where $j_{0,1} \approx 2.405$ is the first zero of the Bessel function $J_0$. This confirms that the positive curvature at the tip acts as an attractive potential, lowering the momentum uncertainty relative to the flat Euclidean case.
		\item \textbf{Large-ball regime ($R \to \infty$).} 
		As the radius increases, the manifold transitions into a cylinder $\mathbb{R} \times S^1$. The radial Laplacian simplifies asymptotically to $\partial_r^2$, effectively reducing the problem to a one-dimensional particle in a box of length $R$. Consequently, the eigenvalue scales as
		\begin{equation}
			\lambda_1(B_R) \;\sim\; \frac{\pi^2}{4R^2} \qquad (R \to \infty).
		\end{equation}
		Remarkably, this asymptotic behavior coincides exactly with the universal Barta-type lower bound derived in Corollary~\ref{cor:barta} (with $\pi/2 \approx 1.57$). This demonstrates that our coordinate-independent lower bound $\sigma_p r \ge \pi\hbar/2$ is asymptotically optimal for geometries that degenerate into lower-dimensional effective spaces, capturing the transition from 2D radial symmetry to 1D longitudinal confinement.
	\end{itemize}
	
	\medskip
	\noindent
	\textit{Injectivity Radius.} Since the manifold is topologically $\mathbb{R}^2$ and the radial function $\tanh(r)$ is strictly monotonic, there are no conjugate points along radial geodesics from the origin. Thus, $\operatorname{inj}_{\Sigma}(0) = \infty$, allowing for arbitrarily large geodesic balls as discussed in the asymptotic analysis.

\paragraph{Example 4.4} [Injectivity radius at the Schwarzschild throat]\\
In general there is no simple closed-form expression for the injectivity-radius of a generic Riemannian \(3\)-manifold, and the maximal Schwarzschild slice is no exception. For the purposes of Assumption~\ref{ass:min}, however, one
does not need an explicit formula for \(\mathrm{inj}_{\Sigma}(x)\) at every point \(x\in\Sigma\). The
assumption only requires that the radius \(r\) of the geodesic ball \(B_{\Sigma}(p,r)\) be strictly
smaller than \(\mathrm{inj}_{\Sigma}(p)\) at the chosen center \(p\). It is therefore natural to look for a point where the injectivity radius is smallest: any lower bound obtained from such a ``worst case'' automatically controls all admissible localization radii on the slice. On the time-symmetric Cauchy slice of the maximal Schwarzschild extension this role is played by the Einstein--Rosen throat \cite{CollasKlein:SchwarzschildWormhole}.

More concretely, in the maximal Schwarzschild extension the time-symmetric Cauchy slice
\((\Sigma,h)\) can be written in warped-product form
\[
h = d\ell^{2} + R(\ell)^{2} d\Omega^{2},
\]
where \(\ell\in\mathbb{R}\) is a proper-distance coordinate and \(R(\ell)\) is the area radius of the
\(2\)-spheres \(\{\ell=\mathrm{const}\}\). The minimal \(2\)-sphere (Einstein--Rosen throat / apparent horizon) is located at \(\ell = \ell_{0}\) with \(R(\ell_{0}) = 2M\) and \(R'(\ell_{0}) = 0\); the corresponding \(2\)-sphere is totally geodesic and carries the round metric of radius
 \[
R_{\min} := 2M.
\]
Let \(p\in\Sigma\) lie on this minimal \(2\)-sphere \(S_{H}\). Geodesics in \((\Sigma,h)\) that are tangent
to \(S_{H}\) at \(p\) remain in \(S_{H}\) and coincide with great-circle geodesics on the round
\(2\)-sphere of radius \(R_{\min}\). These great circles have sectional curvature \(1/R_{\min}^{2}\), and
along them the first conjugate point occurs at distance \(\pi R_{\min}\). Hence the conjugate radius
at \(p\) is
\[
\mathrm{conj}_{\Sigma}(p): = \pi R_{\min} = 2\pi M.
\]
Moreover, the shortest nontrivial closed geodesics through \(p\) are precisely these great circles,
of length \(2\pi R_{\min} = 4\pi M\), so that half their length is again \(2\pi M\). By Klingenberg's
injectivity-radius formula,
\[
\mathrm{inj}_{\Sigma}(p)
= \min \Big\{ \mathrm{conj}_{\Sigma}(p), \tfrac{1}{2}\ell_{p} \Big\}
= 2\pi M,
\]
where \(\ell_{p}\) denotes the length of the shortest nontrivial closed geodesic through \(p\).
Thus the injectivity radius at the Schwarzschild throat is finite and of order \(M\). In particular,
for any \(p\in S_{H}\) the geodesic balls \(B_{\Sigma}(p,r)\) satisfy Assumption~\ref{ass:min} for every radius 
\[0 < r < 2\pi M,\] 
and the value \(2\pi M\) can be regarded as a natural lower bound for the injectivity radius scale relevant to our intrinsic Heisenberg-type estimates. Since this quantity is derived exclusively from intrinsic metric data, it serves as a coordinate-independent lower bound for the localization scale.

\paragraph{Example 4.5} [Isotropic maximum-length model of GUP]\\
As a final example we consider the isotropic maximum-length (ML) model introduced in the context of the GUP \cite{Perivolaropoulos:2017rgq,Parsamehr:2024,Skara:2019,BensalemBouaziz:2022,Lawson:2020}. 
Its kinematics is defined algebraically by deformed position-momentum commutation relations on $\mathbb{R}^3$ with Cartesian coordinates $x_i$,
\begin{equation}
	[\hat{x}_i,\hat{x}_j] = 0, 
	\qquad
	[\hat{x}_i,\hat{P}_j] = i\hbar\,\frac{\delta_{ij}}{1-\alpha r^2},
	\qquad
	r^2 := x_1^2+x_2^2+x_3^2,
\end{equation}
where $\alpha>0$ is a deformation parameter. In this algebraic formulation, the factor $(1-\alpha r^2)^{-1}$ is the only modification of the usual flat-space commutator; the position operators remain mutually commuting.
Applying a metric-reconstruction procedure with the general commutator rule \eqref{eq:CommxP}, one reads off the orthonormal frame fields from the mixed commutator:
\begin{equation}
	X_j(x_i) = \frac{\delta_{ij}}{1-\alpha r^2}
	\quad\Longrightarrow\quad
	X_j = \frac{1}{1-\alpha r^2}\,\frac{\partial}{\partial x_j}.
\end{equation}
Requiring $\{X_j\}$ to be orthonormal determines the inverse metric and hence the configuration-space metric of the ML model,
\begin{equation}
	g_{ij}(x) = (1-\alpha r^2)^2\,\delta_{ij},
\end{equation}
so that the configuration space is the conformally flat Riemannian manifold $(\mathbb{R}^3,g_{\mathrm{ML}})$ with conformal factor $1{-}\alpha r^2$. 
In this way, the ``maximum-length'' deformation of the flat-space commutator is reinterpreted as standard quantum mechanics on the curved manifold $(\mathbb{R}^3,g_{\mathrm{ML}})$, with the entire modification encoded in the conformal factor of the metric, the induced symplectic structure, and the momentum commutator.
The Ricci tensor of $g_{\mathrm{ML}}$ in Cartesian coordinates $(x_1,x_2,x_3)$ is
\begin{equation}
	\operatorname{Ric}_{ij}(x)
	= \frac{8\alpha}{(\alpha r^2-1)^2}
	\bigl(\alpha x_i x_j - (\alpha r^2-1)\delta_{ij}\bigr),
	\qquad r^2 = x_1^2+x_2^2+x_3^2
\end{equation}
and therefore the ML metric satisfies the uniform Ricci lower bound
\begin{equation}
	\operatorname{Ric}_{g_{\mathrm{ML}}}(X,X)
	\;\ge\; 8\alpha\, g_{\mathrm{ML}}(X,X),
	\qquad \forall X\in T_x\mathbb{R}^3,\; r(x)<\alpha^{-1/2}.
\end{equation}
In the standard notation $\operatorname{Ric}\ge (n-1)K\,g$ with $n=3$ this corresponds to $K = 4\alpha$.
This Ricci lower bound, together with nonnegative mean curvature of $\partial B_r$, allows one to apply Yang--Ling type inequalities \cite{Yang1999,Ling2006} to obtain 
\begin{equation}\label{eq:ling}
	\lambda^{\mathrm{ML}}_1
	\;\ge\; \frac{\pi^2}{4r^2} + 4\alpha,
\end{equation}
where we have quoted Ling's stronger bound \cite{Ling2006}. 
By way of comparison, Barta's inequality (Corollary~\ref{cor:barta}) yields
\begin{equation}\label{eq:Barta-ML}
	\lambda^{\mathrm{ML}}_1 \;\ge\; \frac{\pi^2}{4r^2}\,.
\end{equation} 
In particular, inequality~\eqref{eq:Barta-ML} captures the universal
$r^{-2}$-behavior of the first Dirichlet eigenvalue, whereas the refined
estimate~\eqref{eq:ling} improves this bound by an additional strictly positive,
curvature-induced term $4\alpha$, which reflects the presence of an intrinsic spectral gap for the ML model. Moreover, for the ML metric the injectivity radius depends on the base point
$p \in \Sigma$ and satisfies the bounds
\begin{equation}\label{eq:inj-ML-bounds}
	0 \;\le\; \mathrm{inj}_{\mathrm{ML}}(p)
	\;\leq\; \frac{2}{3\sqrt{\alpha}},
	\qquad p \in \Sigma.
\end{equation}
Hence the maximal admissible radius of geodesic balls centered at the origin is
$r < \mathrm{inj}_{\mathrm{ML}}(0) = 2/(3\sqrt{\alpha})$, while the injectivity radius
tends to zero as $p$ approaches the Euclidean boundary $r = 1/\sqrt{\alpha}$ of
the ML configuration space.

Taken together, the comparison inequalities for $\lambda^{\mathrm{ML}}_1$ and the injectivity-radius estimate \eqref{eq:inj-ML-bounds} provide a global control of the spectrum in the maximum-length geometry: for every base point ($p$) and every admissible localization radius ($r<\mathrm{inj}^{\mathrm{ML}}(p)$) the first Dirichlet eigenvalue is bounded from below by an $r^{-2}$- term, and the Yang-Ling refinement adds a strictly positive curvature-induced contribution. From the perspective of the GUP phenomenology, however, the most relevant regime is precisely the opposite one, in which the localization scale is small compared to the curvature scale set by ($\alpha^{-1/2}$), so that the deformation away from flat space is weak.

It is therefore natural to complement these global comparison bounds by a local small-ball analysis around the origin in which the metric $g_{\mathrm{ML}}$ may be treated as a mild conformal perturbation of the Euclidean metric. Specializing to geodesic balls centered at the origin with radius ($r$) satisfying $\alpha r^2\ll 1$ and hence $r<\mathrm{inj}_{\mathrm{ML}}(0))$, the Dirichlet problem becomes accessible to a straightforward radial perturbation calculation. In this perturbative regime one finds the asymptotic expansion
\begin{equation}
	\lambda^{\mathrm{ML}}_1(r)
	= \frac{\pi^2}{r^2} - 4\alpha + O(\alpha^2 r^2),
\end{equation}
so that the leading $r^{-2}$ behaviour coincides with that of a geodesic ball in a three-sphere $S^3$ of constant sectional curvature $4\alpha$; the Ricci term $4\alpha$ is the first curvature correction.

\section{Discussion and conclusion}\label{sec:Diss}

The analysis carried out in this work revisits Heisenberg's uncertainty principle in a preparation-based form that is adapted to curved spacetime. Instead of deforming the canonical commutation relations or postulating a specific generalized uncertainty principle, we have kept standard quantum mechanics on a fixed spacetime background and encoded the influence of gravity in the geometry of the localization region itself. Strict von Neumann--Lüders localization to a geodesic ball $B_\Sigma(p,r)$ on a spacelike hypersurface $(\Sigma,h)$ was modeled by Dirichlet data, and the resulting momentum uncertainty was described by a geometric, coordinate-independent variance built from DeWitt-type momentum operators. In this sense, the familiar single-slit thought experiment is lifted from flat space to arbitrary slices of a Lorentzian spacetime.

From this viewpoint, the main intrinsic estimate can be read as a spectral version of Heisenberg's principle. For any normalized state $\psi$, strict localization in the ball enforces the lower bound
\[
\sigma_p(\psi) \;\geq\; \hbar \,\sqrt{\lambda_1}\,,
\]
where $\lambda_1$ is the first Dirichlet eigenvalue of the Laplace--Beltrami operator on the ball. The bound is intrinsic in a strong sense: lapse, shift and extrinsic curvature never enter, and only the induced Riemannian metric on the slice survives. Physically, $\lambda_1$ condenses into a single number all geometric features of the localization region that matter for the trade-off between position and momentum: its proper radius and volume, the shape and mean curvature of its boundary, and the curvature profile in the interior. The kinetic-energy window separating the gradients of $|\psi|$ and of $\psi$ itself shows that a nontrivial part of the momentum uncertainty is already fixed by the modulus of the wave function, while fluctuations of the phase gradient only add a nonnegative excess. The equality conditions make it clear that the lower bound is sharp but typically not saturated on curved balls, reflecting the fact that a globally constant phase gradient is hard to realize in a nontrivial geometry.

In many situations one does not know $\lambda_1$ explicitly, but one does have qualitative control over the boundary of the ball. Under the weak mean-convexity assumption, the boundary-distance Hardy inequality and the vector-field version of Barta's method lead to universal, scale-invariant product bounds
\[
\sigma_p\, r \;\geq\; \frac{\hbar}{2}\,, 
\qquad
\sigma_p\, r \;\geq\; \frac{\pi\hbar}{2}\,,
\]
for all strictly localized states. These inequalities have several noteworthy features. They depend only on the proper radius of the ball and the sign of the radial Laplacian; the detailed curvature in the interior is irrelevant. The Hardy baseline with constant $1/2$ is never attained, in line with the flat-space Hardy inequality, which underscores the robustness of the bound against small perturbations of the state. The Barta-type improvement effectively extracts everything that can be said from mean-convexity alone; any further sharpening of the constant must bring in additional curvature information, for instance via classical eigenvalue comparison theorems or explicit model calculations on constant-curvature spaces.

The examples in Section~\ref{sec:Examples} illustrate how these general statements play out in concrete geometries. On unimodular group manifolds such as the three-dimensional Heisenberg (Nil) geometry, the geometric potential $V$ vanishes even though the scalar curvature is negative, so the intrinsic momentum uncertainty is governed entirely by the Dirichlet spectrum on geodesic balls. On hyperbolic space $H^3$ the non-unimodular Lie algebra generates a positive constant $V$ that coincides with the spectral gap of the Dirichlet Laplacian on large balls. In this case the intrinsic variance separates neatly into a confinement contribution with the familiar $1/r^2$ scaling and a curvature-induced offset. Witten's cigar provides an example with non-constant curvature: here $V(r)$ is a genuine potential that regularizes the effective Hamiltonian near the tip and decays along the asymptotically cylindrical region. Finally, the maximum-length model from the GUP literature shows that a purely algebraic deformation of the commutator can be reinterpreted as standard quantum mechanics on a conformally flat configuration space. In that rereading, the intrinsic lower bound on $\sigma_p$ is dictated by the Dirichlet spectrum of the reconstructed metric, and the deformation parameter plays the role of a curvature scale.

Conceptually, these observations suggest a different perspective on generalized or extended uncertainty relations. Instead of taking them as a fundamental departure from the canonical commutation relations, one may regard them as effective manifestations of the spectral geometry of the localization regions available in a given gravitational setting. The present framework demonstrates that on any spacelike hypersurface the interplay between localization and momentum can be quantified without leaving standard quantum mechanics on curved spacetime: once the induced Riemannian metric and a localization region are specified, the spectral geometry of that region fixes a lower bound on the momentum standard deviation $\sigma_p$. In this sense, the ``minimal lengths'' and ``maximal localization scales'' that often appear in phenomenological GUP models need not be new constants of nature; they can emerge as geometric scales such as the injectivity radius, curvature radius or throat size of the slice.

At the same time, the limitations of the present treatment point toward several natural extensions. Throughout we worked with test particles on a fixed background, so that backreaction of the localized quantum state on the geometry was ignored. In regimes where the momentum spread demanded by the intrinsic bounds approaches the Planck scale, one expects significant gravitational backreaction and possibly horizon formation; a self-consistent treatment would then have to determine the spectrum of $B_\Sigma(p,r)$ jointly with the Einstein equations. A second direction is the generalization from single-particle quantum mechanics to quantum field theory on curved backgrounds. There the notion of localization is more subtle, but the spectral point of view suggests that analogous bounds might be formulated for suitably smeared field operators or for local algebras associated with geodesic balls. Finally, in the context of quantum sensing, the universal product bounds derived here can be interpreted as hard limits on the spatial resolution of probes that are confined to strongly curved regions such as black-hole throats or cosmological horizons, and may serve as design constraints for measurement schemes that try to balance spatial focus against momentum or energy noise.

To summarize, the results obtained in this work recast Heisenberg's uncertainty principle as a statement about the intrinsic spectral geometry of spacelike hypersurfaces in general relativity. Strict localization to a geodesic ball carries an unavoidable spectral cost, quantified by the first Dirichlet eigenvalue and, under weak convexity assumptions, bounded below by universal constants that are independent of the foliation and insensitive to the detailed interior curvature. These intrinsic Heisenberg-type inequalities show that quantum uncertainty and spacetime geometry are intertwined already at the level of single-particle quantum mechanics on curved backgrounds, well before any explicit quantization of the gravitational field is invoked.

Ultimately, this implies that the uncertainty principle ceases to be merely a kinematic restriction of the state vector, but reveals itself as a geometric necessity enforced by the curvature of spacetime itself.

\section*{Acknowledgments}
	The author thanks the spectral geometry community for many classical tools; any errors are the author's responsibility.
	\appendix
	
	\section{Integration-by-parts identity for the DeWitt momenta}\label{app:DeWitt-IBP}
	
	For completeness we sketch the integration-by-parts computation leading to the quadratic-form
	identity \eqref{eq:quadratic-form-identity}. Let $\{X_a\}_{a=1}^n$ be an $h$-orthonormal frame
	on $B_\Sigma(p,r)$, and write $f_a := \mathrm{div}_h X_a$. On scalar wave functions the
	covariant derivative reduces to directional differentiation, $\nabla^h_{X_a}\psi = X_a\psi$, so
	that
	\[
	P_a\psi = - i\hbar\left(X_a\psi + \frac{1}{2} f_a \psi\right), \qquad
	a = 1,\dots,n.
	\]
	For $\psi \in C^\infty_0(B_\Sigma(p,r))$ we have
	\begin{align*}
		\langle P^2\rangle
		&= \sum_{a=1}^n \|P_a\psi\|_{L^2}^2
		= \hbar^2 \sum_{a=1}^n \int_{B_\Sigma(p,r)}
		\bigl|X_a\psi + \tfrac{1}{2} f_a \psi\bigr|^2\,\mathrm{d}\mu_h \\
		&= \hbar^2 \sum_{a=1}^n \int_{B_\Sigma(p,r)}
		\Bigl( |X_a\psi|^2
		+ \frac{1}{4} f_a^2 |\psi|^2
		+ \frac{1}{2} f_a\,X_a(|\psi|^2)
		\Bigr)\,\mathrm{d}\mu_h,
	\end{align*}
	where we used
	$X_a(|\psi|^2) = X_a\psi\,\overline{\psi} + \psi\,X_a\overline{\psi}$ for the cross term. %
	To handle the last term, note that for any smooth scalar field $\varphi$ the divergence identity
	\[
	\mathrm{div}_h(\varphi X_a) = X_a\varphi + f_a \varphi
	\]
	implies, upon choosing $\varphi = |\psi|^2$ and integrating over $B_\Sigma(p,r)$, that
	\[
	\int_{B_\Sigma(p,r)} X_a(|\psi|^2)\,\mathrm{d}\mu_h
	= - \int_{B_\Sigma(p,r)} f_a\,|\psi|^2\,\mathrm{d}\mu_h,
	\]
	because the boundary term vanishes for Dirichlet data ($\psi|_{\partial B_\Sigma(p,r)} = 0$). Applying the same identity to $\varphi = f_a |\psi|^2$ yields
	\[
	\int_{B_\Sigma(p,r)} \bigl( X_a f_a \,|\psi|^2
	+ f_a\,X_a(|\psi|^2)
	+ f_a^2 |\psi|^2
	\bigr)\,\mathrm{d}\mu_h = 0,
	\]
	hence
	\[
	\int_{B_\Sigma(p,r)} f_a\,X_a(|\psi|^2)\,\mathrm{d}\mu_h
	= - \int_{B_\Sigma(p,r)} \bigl(X_a f_a + f_a^2\bigr)\,|\psi|^2\,\mathrm{d}\mu_h.
	\]
	Substituting this back into the expression for $\|P\|_h^2[\psi]$ gives, for each $a$,
	\[
	\int_{B_\Sigma(p,r)}
	\left( \frac{1}{4} f_a^2 |\psi|^2
	+ \frac{1}{2} f_a\,X_a(|\psi|^2)
	\right)\,\mathrm{d}\mu_h
	= - \int_{B_\Sigma(p,r)}
	\left( \frac{1}{2} X_a f_a + \frac{1}{4} f_a^2 \right)
	|\psi|^2\,\mathrm{d}\mu_h.
	\]
	Summing over $a$ and using the pointwise identity
	$\sum_{a=1}^n |X_a\psi|^2 = \|\nabla^h\psi\|_h^2$, we obtain
	\[
	\langle P^2\rangle
	= \hbar^2 \int_{B_\Sigma(p,r)} \|\nabla^h\psi\|_h^2\,\mathrm{d}\mu_h
	- \hbar^2 \int_{B_\Sigma(p,r)}\! V(x)\,|\psi|^2\,\mathrm{d}\mu_h,
	\]
	with $V$ as in \eqref{eq:Vpotential}. Since $C^\infty_0(B_\Sigma(p,r))$ is dense in
	$H^1_0(B_\Sigma(p,r))$ and both sides define continuous quadratic forms on
	$H^1_0(B_\Sigma(p,r))$, the identity extends by continuity to all
	$\psi \in H^1_0(B_\Sigma(p,r))$, which is exactly
	\eqref{eq:quadratic-form-identity}.

\end{document}